\documentclass[adraft,copyright]{eptcs}
\providecommand{\event}{ITRS 2010}
\usepackage{proof,bussproofs}
\usepackage{latexsym}
\usepackage{amsmath}
\usepackage{breakurl}
\usepackage{amssymb}
\usepackage{xcolor}
\newcommand\inters{\cap}
\newcommand\aconj{\&}
\newcommand\sconj{\cap}

\newcommand\ra{\rightarrow}

\newcommand{\xspace}{ \makebox[1em]{} }

\newcommand\LJ{\textbf{LJ}}

\newcommand\IT{\textbf{IT}}
\newcommand\ITD{\vdash_{\mathrm{IT}}}

\newcommand\IL{\textbf{IL}}

\newcommand\ISL{\textbf{ISL}}
\newcommand\ISLD{\vdash_{\mathbf{ISL}}}
\newcommand\ISC{\textbf{ISC}}

\newcommand\IUSC{\textbf{IUSC}}
\newcommand\ISCD{\vdash_{\mathbf{ISC}}}

\newcommand\molecule[2]{[(#1;#2)\mid i\in I]}

\newcommand\mM{{\mathcal M}}
\newcommand\mN{{\mathcal N}}

\def\ax{\mbox{\bf (Ax)}}
\def\FUS{\mbox{\bf (Fus)}}
\def\Pru{\mbox{\bf (P)}}
\def\cut{\mbox{\bf (cut)}}

\def\weak{\mbox{\bf (W)}}
\def\exch{\mbox{\bf (X)}}
\def\weakex{\mbox{\bf(WX)}}

\def\excon{\mbox{\bf (XC)}}
\def\contr{\mbox{\bf (C)}}
\def\impL{\mbox{\bf  ($\rightarrow$L)}}
\def\impR{\mbox{\bf  ($\rightarrow$R)}}

\def\andL{\mbox{\bf ($\aconj$ L)}}

\def\andR{\mbox{\bf ($\aconj$R)}}
\def\intL{\mbox{\bf ($\sconj$L)}}

\def\intLk{\mbox{\bf ($\sconj$L$_k$)}}

\def\intR{\mbox{\bf ($\sconj$R)}}

\def\intI{\mbox{\bf ($\sconj$I)}}
\def\intE{\mbox{\bf ($\sconj$E)}}
\def\intEk{\mbox{\bf ($\sconj$E$_k$)}}
\def\impE{\mbox{\bf ($\rw$E)}}
\def\impI{\mbox{\bf ($\rw$I)}}
\def\andI{\mbox{\bf ($\aconj$I)}}
\def\andEk{\mbox{\bf ($\aconj$E$_k$)}}

\newenvironment{proof}{\noindent {\bf Proof}\quad}{\nobreak
\quad\nobreak\vrule height5pt width3pt depth2pt}

\newcommand{\gra}{\alpha}
\newcommand{\grb}{\beta}
\newcommand{\grg}{\gamma}

\newcommand{\grtt}{\theta}

\newcommand{\grm}{\mu}
\newcommand{\grn}{\nu}

\newcommand{\grr}{\rho}
\newcommand{\grs}{\ensuremath{\sigma}}
\newcommand{\grt}{\tau}

\newcommand{\GG}{\ensuremath{\Gamma}}
\newcommand{\DD}{\Delta}

\newcommand{\rw}{\rightarrow}

\definecolor{mygreen}{rgb}{0.3,0.9,0.6}

\newtheorem{thm}{\bfseries Theorem}[section]   %%%\begin{thm}...\end{thm}
\newtheorem{lem}[thm]{\bfseries Lemma}
\newtheorem{dfn}[thm]{\bfseries Definition}

\newtheorem{cor}[thm]{\bfseries Corollary}
\newtheorem{exam}[thm]{\bfseries Example}

\def\LJax{\mbox{\scriptsize\bf (Ax)}}
\def\LJcut{\mbox{\scriptsize\bf (cut)}}
\def\LJweak{\mbox{\scriptsize\bf (W)}}

\def\LJand{\wedge}
\def\LJimpL{\mbox{\scriptsize\bf($\rightarrow$L)}}

\def\LJandL{\mbox{\scriptsize\bf($\LJand$L)}}

\def\LJandR{\mbox{\scriptsize\bf($\LJand$R)}}

\newtheorem{exmplnt}[thm]{\bfseries Example}

\newtheorem{rmrknt}[thm]{\bfseries Remark}

\long\def\hide#1\endhide{}

\usepackage{multicol}
\title{Intersection Logic in sequent calculus style}
\author{Simona Ronchi Della Rocca
\institute{Dipartimento di Informatica\\
Universita di Torino\\
IT-10149 Torino, Italy}
\email{ronchi@di.unito.it}
\and
Alexis Saurin
\institute{Laboratoire CNRS PPS \& INRIA pi.r2\\
Paris, France}
\email{alexis.saurin@pps.jussieu.fr}
\and
Yiorgos Stavrinos \qquad\qquad\qquad Anastasia Veneti
\institute{Graduate Program in Logic and Algorithms (MPLA)\\
Department of Mathematics\\
University of Athens\\
GR-15784 Zografou, Greece}
\email{g.stavrinos@math.ntua.gr \qquad\qquad\qquad tassiana98@gmail.com}
}

\def\titlerunning{Intersection Logic in sequent calculus style}
\def\authorrunning{S.\ Ronchi Della Rocca, A.\ Saurin, et al.} %%% Y.\ Stavrinos, A.\ Veneti}

\EnableBpAbbreviations

\begin{document}
\maketitle
\section{Introduction}
The intersection type assignment system \IT\ (shown in Figure \ref{ITrules}) is a deductive system that assigns formulae
(built from the intuitionistic implication $\rightarrow$ and the
intersection $\sconj$) as types to the untyped $\lambda$-calculus. It has been defined by Coppo and Dezani \cite{it},
in order to increase the typability power of the simple type assignment system.
In fact, \IT\ has a strong typability power, since it gives types to all the strongly normalizing terms \cite{pott80,krivine}.
Intersection types, supported by a universal type and a suitable
pre-order relation, has been used for describing the $\lambda$-calculus semantics in different domains
(e.g. Scott domains \cite{BCD, CDHL}, coherence spaces \cite{Paolini}), allowing to characterize interesting semantical notions
like solvability and normalizability, both in call-by-name and call-by-value settings \cite{ronchi2}.

Differently from other well known type assignments for $\lambda$-calculus, like for example
the simple or the second order one, \IT\ has not been designed as decoration of a logical system, following the so called
Curry-Howard isomorphism, relating in particular types with logical formulae.
In fact, despite to the fact that it looks like a decoration of the implicative and additive fragment of the intuitionistic logic,
this is not true. Indeed, while the connective $\ra$ has the same behaviour of the intuitionistic implication, the intersection
$\sconj$ does not correspond to the intuitionistic conjunction, as Hindley pointed out in \cite{hin}, since the rule introducing it
has the meta-theoretical condition that the two proofs of the premises be decorated by the same term,
and this constraint is no more explicit when terms are erased.

It is natural to ask if there is a logic with a natural correspondence with intersection types. Some attempts have been made,
which will be briefly recalled in Section \ref{Other}.
Our work follow the line started from ~\cite{il} and~\cite{isl,islrev}, where the authors defined {\em Intersection Logic} $\IL$
and {\em Intersection Synchronous Logic} $\ISL$, respectively.
%The problem of a logical foundation of intersection
%types has been proposed by Hindley~\cite{hin} as a consequence of the observation that the intersection connective does not
%coincide with the intuitionistic conjunction.
$\ISL$ is a deductive system in natural deduction style proving {\em molecules} which
are finite multisets of {\em atoms} which, in turn, are pairs of a context and a formula. The notion of molecule allows to distinguish
between two sets of connectives, {\em global} and {\em local}, depending on whether they are introduced and eliminated in parallel on all
atoms of a molecule or not.
Note that global and local connectives are called asynchronous and synchronous, respectively, in~\cite{isl,islrev}.
In particular, two kinds of conjunction are present: the global one corresponding to intuitionistic conjunction and the local one
corresponding to intersection. Then, the intersection type assignment system is obtained by decorating $\ISL$ with $\lambda$-terms
in the Curry-Howard style. Roughly speaking, a molecule can be mapped into a set of intuitionistic proofs which are {\em isomorphic},
in the sense that some rules are applied in parallel to all the elements in the set. The isomorphism is obtained by collapsing the two
conjunctions.

$\ISL$ put into evidence the fact that the usual intuitionistic conjunction can be splitted into two different connectives,
with a different behaviour,
and the intersection corresponds to one of it.
So we think that $\ISL$ is a logic that can be interesting in itself, and so we would like to explore it from a proof-theoretical point of
view. To this aim, we present here a sequent calculus formulation of it, that we call $\ISC$, and we prove that it enjoys the property
of cut elimination. The proof is not trivial, since the play between global and local connectives is complicated and
new structural rules are needed.

Stavrinos and Veneti~\cite{tut,cie} further enhanced $\IL$ and $\ISL$ with a new connective: the logical counterpart of the union
operator on types~\cite{bar}. It turns out that union is a different kind of disjunction. In the extensions of $\IL$ and $\ISL$
presented in~\cite{tut}, though, the status of union is not so clear; in particular, union is not an explicitly local version of intuitionistic
disjunction, since its elimination rule involves a global side-effect. So, the pair union-disjunction does not enjoy the nice characterization
of the pair intersection-conjunction, which are the local and global versions, respectively, of the same connective. Moving to a sequent
calculus exposition \cite{cie} allows for a very nice and symmetric presentation of both intersection and union, which are
now the local versions of intuitionistic conjunction and disjunction, respectively.
Our next step will be to do a proof theoretical investigation of this new calculus.

\medskip
%An alternative cut-elimination proof can be found in draft version in~\cite{ced}. We use a direct approach following the lines of Gentzen's
%Hauptsatz~\cite{gir}. We  first prove elimination of a topmost maximal-degree ``multicut'' which is, in short, a generalized cut. Generalizing
%to multicuts is necessitated by the presence of structural rules, notwithstanding that a multicut is not considered to be an actual rule of
%the system. We then prove elimination of all maximal-degree cuts from a derivation, hence lowering the degree of a derivation. This is done
%by repeating a finite number of times the procedure ``elimination of topmost maximal-degree cut'' which is a special case of the generalized
%procedure described above. Finally, we prove cut-elimination by repeating a finite number of times the procedure ``lowering the degree of a
%derivation'' to reach a derivation of zero degree, i.e.\ a cut-free derivation.
The paper is organized as follows. In Section 2 the system $\ISL$ is recalled, and its sequent calculus version $\ISC$ is presented. In Section 3 the relation between $\ISC$ and Intuitionistic Sequent Calculus $\LJ$ is shown. Section 4 contains a brief survey of the cut-elimination procedure
in $\LJ$. In Section 5 the cut elimination steps of $\ISC$ are defined, and in Section 6 the cut-elimination algorithm is given.
Moreover Section 7 contains a brief survey of the other approaches to the problem of a logical foundation of intersection types.

\begin{figure}\label{ITrules}
\begin{center}
$$ \infer[(A)]{\Gamma, x:\sigma \ITD x:\sigma}{} $$
%\xspace & \xspace $
%\infer[(W)]{\Gamma, \tau \NJD \sigma}
%{\Gamma \NJD \sigma}  $
%\xspace & \xspace
%$ \infer[(X)] {\Gamma_{1},\sigma_{2},\sigma_{1},\Gamma_{2} \NJD \sigma}
%{\Gamma_{1},\sigma_{1},\sigma_{2},\Gamma_{2}\NJD
%\sigma} $
\\

\begin{tabular}{cc}
$\infer[(\sconj I)]{\Gamma \ITD M: \sigma\sconj \tau} {\Gamma \ITD M:\sigma \quad \Gamma
\ITD M: \tau}  $
\xspace & \xspace
$\infer[(\sconj E_{k}) k=L,R] {\Gamma \ITD
 M: \sigma_{k}} {\Gamma \ITD   M: \sigma_{L}\sconj \sigma_{R}}  $
%\xspace & \xspace $\infer[(\sconj E^r)] {\Gamma\NJD  M: \tau} {\Gamma \NJD
% M: \sigma\sconj \tau}  $
\\
\\
$ \infer[(\ra I)] {\Gamma \ITD  \lambda x.M: \sigma
\ra \tau}  {\Gamma, x:\sigma  \ITD M: \tau} $ \xspace & \xspace
$\infer[(\ra E)] {\Gamma  \ITD
MN:\tau}{\Gamma \ITD M:\sigma \ra \tau \quad \Gamma \ITD N:\tau}   $
\end{tabular}
\end{center}
\caption{The intersection type assignment system \IT.}
\end{figure}

\section{The system $\ISC$}
In this section the system $\ISL$ presented
in~\cite{isl,islrev} is recalled. Moreover, an equivalent sequent calculus version is given.

\begin{dfn}[ISL]
\begin{itemize}
\item [i)]\emph{Formulas} are generated by the grammar
\[\grs::=a\:|\:\grs\ra\grs\:|\:\grs\aconj\grs\:|\:\grs\sconj\grs\] where
$a$ belongs to a countable set of propositional variables.
Formulas are ranged over by small greek letters.
\item[ii)] \emph{Contexts} are finite sequences of formulas, ranged over by $\,\GG,\DD,E,Z$.
The \emph{cardinality} $|\GG|$ of a context $\GG$ is the
number of formulas in $\GG$.
\item[iii)] \emph{Atoms} $(\GG;\grs)$ are pairs of a context and a formula, ranged over by $\,\mathcal{A,B}$.
\item[iv)] \emph{Molecules}
$\,[(\gra_1^i,\ldots\gra_m^i;\gra_i)\:|\:1\leqslant i\leqslant n]\,$ or
$\,[(\gra_1^i,\ldots\gra_m^i;\gra_i)\:|\:i\in I]\,$ or $\,[(\gra_1^i,\ldots\gra_m^i;\gra_i)]_i$
are finite multisets of atoms such that all atoms have the same context cardinality.
Molecules are ranged over by $\mM, \mN$.
\item[v)] $\ISL$ is a logical system proving molecules in natural deduction
style. Its rules are displayed in Figure~\ref{iund}.
We write $\ISLD\mM$ when there is a deduction proving the molecule $\mM$, and
$\pi:\mM$ when we want to denote a particular derivation $\pi$ proving $\mM$.
\end{itemize}
\end{dfn}

%\begin{thm}
%(\cite{}) $\IUND$ is strongly normalizing.         %%% 31/3/2010,yiorgos: citation? there is no proof yet
%\end{thm}

Some comments are in order. Rule $(P)$ is derivable and it is useful for the normalization procedure.
The connectives $\ra,\aconj$ are {\em global}, in the sense that they are both
introduced and eliminated in all the atoms of a molecule by a unique rule, while
the connective $\sconj$ is {\em local}, since it is introduced and eliminated in just one atom of a molecule.
%The connective $\sdis$, though, as noted by Veneti and Stavrinos in~\cite{cie}, has a behaviour which is not completely clear: it is local, but the rule eliminating it needs a side condition. Clearly this reflects the not so ``nice'' rule of $\vee$ elimination in natural deduction style. In order to clarify the status of union, we will trasform the system in sequent calculus style.

\begin{figure}
\begin{center}
\footnotesize
\begin{tabular}{c}
$\infer[\scriptsize\ax]{\molecule{\grs_{i}}{\grs_{i}}}{}$\qquad
$\infer[\scriptsize\Pru] {\mM} {\mM\cup \mN}$
\\ \\
$\infer[\scriptsize\weak]{\molecule{\Gamma_i,\grs_i}{ \grt_i}}
{\molecule{\Gamma_i}{ \grt_i}}$ \qquad
$\infer[\scriptsize\exch]{\molecule{\Gamma^{i}_{1} ,\grs_i,\grt_i,\Gamma^{i}_{2}}{\grr_i}}
{\molecule{\Gamma^{i}_{1}
,\grt_i,\grs_i,\Gamma^{i}_{2}}{\grr_i}}$
\\ \\
$ \infer[\scriptsize\impI] { \molecule{\Gamma_i}{\grs_i\ra \grt_i} }
{ \molecule{\Gamma_i,\grs_i}{\grt_i}} $
\qquad
$ \infer[\scriptsize\impE] {\molecule{\Gamma_i}{\grt_i} }
{\molecule{\Gamma_i}{\grs_i\ra\grt_i}\quad\molecule{\Gamma_i}{\grs_i} } $
\\ \\
$ \infer[\scriptsize\andI] {\molecule {\Gamma_i}{\grs_i\aconj \grt_i}}
{\molecule {\Gamma_i}{\grs_i}\quad\molecule{\Gamma_i}{\grt_i}} $
\qquad
$\infer[\scriptsize\andEk,\ k=L,R]
{\molecule{\Gamma_i}{\grs_{i}^{k}}}
{\molecule{\Gamma_i}{\grs_{i}^{L}\aconj\grs_{i}^{R}} }$\qquad
\\ \\
$ \infer[\scriptsize\intI] {\mM\cup[(\Gamma;\grs\inters\grt)]}
{\mM\cup[(\Gamma;\grs),(\Gamma;\grt)]} $\qquad
$ \infer[\scriptsize\intEk,\ k=L,R] { \mM\cup[(\Gamma;\grs_{k})]}
{\mM\cup[(\Gamma;\grs_{L}\sconj\grs_{R})]} $
\end{tabular}
\end{center}\smallskip
\normalsize
\caption{The system $\ISL$}\label{iund}
\end{figure}

The following theorem recall the correspondence between $\ISL$ and $\IT$, proved in \cite{islrev}.
In the text of the theorem, we use the notation $(\Gamma)^{s}$ for denoting a decoration of the context
$\Gamma$ through a sequence $s$ of different variables. More precisely, if $\Gamma = \alpha_{1},..., \alpha_{n}$
and $s=x_{1},...,x_{n}$, $(\Gamma)^{s}$ is the set of assignments
$\{x_{1}:\alpha_{1},...,  x_{n}:\alpha_{n}\}$.

\begin{thm}[\ISL\ and \IT]\label{chiso}
\begin{itemize}
\item[i)] Let
$\ISLD [(\Gamma_1;\alpha_1),\ldots,(\Gamma_{m};\alpha_{m})]$, where
$\alpha_{i} $ and all formulae in $\Gamma_{i} $ do not contain the connective $\aconj$.
Then for every $s$, there is $M$ such that $(\Gamma_i)^s \ITD M:\alpha_i.$
\item[ii)]
If
$(\Gamma_i)^{x_1,\ldots,x_n}\ITD M:\alpha_{i}$ ($i \in I$), then
$\ISLD [(\Gamma_{i};\alpha_{i}) \mid i \in I]$.
\end{itemize}
\end{thm}
Now we will define an equivalent formulation in sequent calculus style.
\begin{dfn}
$\ISC$ is a logical system proving molecules in sequent calculus
style. Its rules are displayed in Figure~\ref{IUSCrules}.
We write $\ISCD\mM$ when there is a derivation proving molecule $\mM$ and
$\pi:\mM$ when we want to denote a particular derivation $\pi$ proving $\mM$.
\end{dfn}

\begin{figure}
\footnotesize

{\bf Identity Rules:} both {\it global}\medskip

$$
\infer[\scriptsize\mbox{\bf (Ax)}]
{[ (\alpha_i ; \alpha_i )| i \in I ]}
{}
\qquad
\infer[\scriptsize\mbox{\bf (cut)}]
{[(\Gamma_i, \Delta_i; \beta_i) |  i \in I ]}
{[ (\Gamma_i; \alpha_i) | i \in I] \quad
[ (\Delta_i, \alpha_i; \beta_i) | i \in I]}
$$\smallskip

{\bf Structural Rules}: {\bf (W)}, {\bf (X)}, {\bf (C)} are {\it global} and  \FUS, {\bf (P)}
are {\it local}\medskip
$$
\infer[\scriptsize\mbox{\bf (W)}]
{[(\Gamma_i , \alpha_i ; \beta_i) | i  \in I]}
{[ (\Gamma_i ; \beta_i) | i \in I]}
\qquad
\infer[\scriptsize\mbox{\bf (X)}]
{[(\Gamma_i, \alpha_i , \beta_i ,  \Delta_i ; \gamma_i )| i \in I]}
{[(\Gamma_i , \beta_i,\alpha_i,\Delta_i ; \gamma_i )| i \in I]
}
\qquad
\infer[\scriptsize\mbox{\bf (C)}]
{[(\Gamma_i , \alpha_i ; \beta_i) | i  \in I]}
{[(\Gamma_i , \alpha_i , \alpha_i ; \beta_i) | i \in I]}
$$

$$\infer[\scriptsize\FUS]{[(\Gamma; \beta),(\Gamma; \beta)]\cup \mM}{[(\Gamma; \beta)]\cup \mM}
\qquad
\infer[\scriptsize\mbox{\bf (P)}] {\mM} {\mM\cup \mN}
$$\smallskip

{\bf Logical Rules:} {\bf $\ra$}, {\bf $\aconj$} are {\it global} and {\bf $\sconj$} is {\it local}\medskip

$$
\infer[\scriptsize\mbox{\bf  ($\rightarrow$L)}]
{[(\Gamma_i,\Delta_i,\alpha_i \rightarrow \beta_i ;  \gamma_i )| i \in I]}
{[(\Gamma_i  ; \alpha_i )| i \in I ] \quad
  [(\Delta_i,\beta_i ; \gamma_i) | i \in I ]}
\qquad
\infer[\scriptsize\mbox{\bf ($\rightarrow$R)}]
{[(\Gamma_i ; \alpha_i  \rightarrow \beta_i )| i \in I]}
{[ (\Gamma_i , \alpha_i ; \beta_i )| i  \in I ]}
$$

$$
\infer[\scriptsize\mbox{\bf ($\aconj$ L)}]
{[(\Gamma_i,\alpha^L_i \aconj \alpha^R_i ; \beta_i) | i \in I]}
{[ (\Gamma_i,\alpha^L_i,\alpha^R_i ; \beta_i) | i \in  I ]}
\qquad
\infer[\scriptsize\mbox{\bf ($\aconj$ R)}]
{[(\Gamma_i ,\Delta_i;  \alpha_i \aconj \beta_i) | i \in I]}
{[ (\Gamma_i; \alpha_i) | i \in I ]
\quad [ (\Delta_i; \beta_i )| i \in I]
}
$$

$$
\infer[\scriptsize\mbox{\bf($\sconj$ L$_k$)},\ k=L,R]
{[(\Gamma, \alpha_L \sconj \alpha_R ; \beta)] \cup  \mM}
{[(\Gamma, \alpha_k ; \beta)] \cup \mM}
\qquad
\infer[\scriptsize\mbox{\bf ($\sconj$ R)}]
{[(\Gamma ; \alpha \cap \beta)]  \cup \mM}
{[(\Gamma ; \alpha),(\Gamma ; \beta)] \cup \mM}
$$

\smallskip
\caption{The system $\ISC$}\label{IUSCrules}
\normalsize
\end{figure}

Notice that the connectives $\ra, \aconj$ are dealt with rules having a multiplicative behaviour of contexts, but the structural
rules ensure an equivalent additive presentation. The connective $\sconj$, though, need to have an additive presentation.
By abuse of notation, we will call global the rules dealing with $\ra$ and $\aconj$ and local the rules dealing with
$\sconj$.
The rules $(\scriptsize \bf P)$ and $\FUS$ are derivable; they will be useful in the cut-elimination procedure.

\begin{lem}\label{clean-lem}
Let $\pi : \ISCD \mM$. There is a function {\em clean}, such that $clean(\pi)$ is a derivation proving $\mM$
which does not contain applications of the rules $(\scriptsize \bf P)$ and $\FUS$.
\end{lem}

\begin{proof}
It is easy to check that both $(\scriptsize \bf P)$ and $\FUS$ commute upwards with any rule preceding them and that they disappear
when following an axiom rule. For instance:
$$\infer[\scriptsize\mbox{\bf (Fus)}]{[(\Gamma; \alpha \cap \beta),(\Gamma;\alpha \cap \beta)]\cup \mM}
{\infer[\scriptsize\mbox{\bf ($\sconj$ R)}]
{[(\Gamma ; \alpha \cap \beta)]  \cup \mM}
{\infer{[(\Gamma ; \alpha),(\Gamma ; \beta)] \cup \mM}{D}}}
\qquad\longrightarrow\qquad
\infer[\scriptsize\mbox{\bf ($\sconj$ R)}]{[(\Gamma; \alpha \cap \beta),(\Gamma;\alpha \cap \beta)]\cup \mM}
{\infer[\scriptsize\mbox{\bf ($\sconj$ R)}]
{[(\Gamma; \alpha), (\Gamma; \beta), (\Gamma ; \alpha \cap \beta)]  \cup \mM}
{\infer[\scriptsize\mbox{\bf (Fus)}]
{[(\Gamma ; \alpha),(\Gamma ; \beta), (\Gamma ; \alpha),(\Gamma ; \beta)] \cup \mM}
{\infer[\scriptsize\mbox{\bf (Fus)}]
{[(\Gamma ; \alpha),(\Gamma ; \beta),(\Gamma ; \beta)] \cup \mM}
{\infer{[(\Gamma ; \alpha),(\Gamma ; \beta)] \cup \mM}{D}}}}}
$$\smallskip
$$
\infer[\scriptsize\mbox{\bf (Fus)}]{[(\beta; \beta),(\beta;\beta)]\cup[ (\alpha_i ; \alpha_i )| i \in I ]}{\infer[\scriptsize\mbox{\bf (Ax)}]
{[(\beta; \beta)]\cup[ (\alpha_i ; \alpha_i )| i \in I ]}
{}}\qquad\longrightarrow\qquad
\infer[\scriptsize\mbox{\bf (Ax)}]
{[(\beta; \beta),(\beta; \beta)]\cup[ (\alpha_i ; \alpha_i )| i \in I ]}
{}
$$\medskip
and
$$
\infer[\scriptsize\mbox{\bf (P)}]{[ (\alpha_i ; \alpha_i )| i \in I ]}{\infer[\scriptsize\mbox{\bf (Ax)}]
{[ (\alpha_i ; \alpha_i )| i \in I ]\cup[ (\alpha_i ; \alpha_i )| i \in J ]}
{}}\qquad\longrightarrow\qquad
\infer[\scriptsize\mbox{\bf (Ax)}]
{[ (\alpha_i ; \alpha_i )| i \in I ]}
{}
$$\end{proof}

We will call {\em clean} a derivation without applications of rules $(\scriptsize \bf P)$ and $\FUS $.
\begin{thm}
$\:\ISLD \mM$ if and only if $\ \ISCD \mM$
\end{thm}
\begin{proof}

{\em \bf (only if)}.
 By induction on the natural deduction style derivation.
Rules {\small\ax,\,\weak,\,\exch}\  are the same in both styles. Rule $(\scriptsize \bf P)$ is derivable in both styles
(see previous lemma and Theorem 11 in \cite{isl}). In case of the global rules, the proof is quite similar to the standard
proof of the equivalence between natural deduction and sequent calculus for intuitionistic case, just
putting similar cases in parallel. So we will show just the case of implication as example, and then the case for local conjunction elimination
(the introduction is the same in both the systems).

Case \impE:\smallskip

{\small
\begin{center}
\AXC{$[(\GG_i\,;\,\gra_i\rw\grb_i)]_i$}
\AXC{$[(\GG_i\,;\,\gra_i)]_i$}
\AXC{$\ \ $}\RL{\scriptsize\ax}
\UIC{$[(\grb_i\,;\,\grb_i)]_i$}\RL{\scriptsize\impL}
\BIC{$[(\GG_i,\gra_i\rw\grb_i\,;\,\grb_i)]_i$}\RL{\scriptsize\cut}
\BIC{$[(\GG_i,\GG_i\,;\,\grb_i)]_i$}\RL{\scriptsize\excon}\dashedLine
\UIC{$[(\GG_i\,;\,\grb_i)]_i$}\DP
\end{center}}\medskip

where the dashed line named $(\bf XC)$ denotes a sequence of applications of rules ($\scriptsize \bf X$) and ($\scriptsize \bf C$).

Case \intEk:\smallskip

{\small
\begin{center}
\AXC{$[(\GG_i\,;\,\grg_i)]_i,(\GG\,;\,\gra_L\cap\gra_R)$}
\AXC{$\ \ $}\RL{\scriptsize\ax}
\UIC{$[(\grg_i\,;\,\grg_i)]_i,(\gra_k\,;\,\gra_k)$}
\RL{\scriptsize\intLk}
\UIC{$[(\grg_i\,;\,\grg_i)]_i,(\gra_L\cap\gra_R\,;\,\gra_k)$}
\RL{\scriptsize\cut}\BIC{$[(\GG_i\,;\,\grg_i)]_i,(\GG\,;\,\gra_k)$}\DP
\end{center}}\medskip

{\em \bf (if)}.

By induction on the sequent calculus style derivation.

We will show just the case of implication, local conjunction elimination and cut.

Case \impL: Let $Z_i=\GG_i,\DD_i,\gra_i\rw\grb_i$. Then:\smallskip

{\small
\begin{center}
\AXC{$[(\DD_i,\grb_i\,;\,\grg_i)]_i$}
\RL{\scriptsize\weakex}\dashedLine
\UIC{$[(Z_i,\grb_i\,;\,\grg_i)]_i$}
\RL{\scriptsize\impI}
\UIC{$[(Z_i\,;\,\grb_i\rw\grg_i)]_i$}
\AXC{$\ \ $}\RL{\scriptsize\ax}
\UIC{$[(\gra_i\rw\grb_i\,;\,\gra_i\rw\grb_i)]_i$}
\RL{\scriptsize\weakex}\dashedLine
\UIC{$[(Z_i\,;\,\gra_i\rw\grb_i)]_i$}
\AXC{$[(\GG_i\,;\,\gra_i)]_i$}
\RL{\scriptsize\weak}\dashedLine
\UIC{$[(Z_i\,;\,\gra_i)]_i$}
\RL{\scriptsize\impE}
\def\defaultHypSeparation{\hskip .1in}
\BIC{$[(Z_i\,;\,\grb_i)]_i$}
\RL{\scriptsize\impE}
\BIC{$[(\underbrace{\GG_i,\DD_i,\gra_i\rw\grb_i}_{Z_i}\,;\,\grg_i)]_i$}\DP
\end{center}}\smallskip

Case \intLk: Let $\GG_i=\GG'_i,\grtt_i$ and $E_i=\GG'_i,\grtt_i,\grtt_i=\GG_i,\grtt_i$. Then:\smallskip

{\small
\begin{center}
\def\ScoreOverhang{1pt}
\AXC{$[(\GG_i\,;\,\grg_i)]_i,(\GG,\gra_k\,;\,\grb)$}
\RL{\scriptsize\weakex}\dashedLine
\UIC{$[(E_i\,;\,\grg_i)]_i,(\GG,\gra_L\cap\gra_R,\gra_k\,;\,\grb)$}
\RL{\scriptsize\impI}
\UIC{$[(\GG_i\,;\,\grtt_i\rw\grg_i)]_i,(\GG,\gra_L\cap\gra_R\,;\,\gra_k\rw\grb)$}
\AXC{$\ \ $}\RL{\scriptsize\ax}
\UIC{$[(\grtt_i\,;\,\grtt_i)]_i,(\gra_L\cap\gra_R\,;\,\gra_L\cap\gra_R)$}
\RL{\scriptsize\weakex}\dashedLine
\UIC{$[(\GG_i\,;\,\grtt_i)]_i,(\GG,\gra_L\cap\gra_R\,;\,\gra_L\cap\gra_R)$}
\RL{\scriptsize\intEk}
\UIC{$[(\GG_i\,;\,\grtt_i)]_i,(\GG,\gra_L\cap\gra_R\,;\,\gra_k)$}
\RL{\scriptsize\impE}
\def\defaultHypSeparation{\hskip .19in}
\BIC{$[(\GG_i\,;\,\grg_i)]_i,(\GG,\gra_L\cap\gra_R\,;\,\grb)$}\DP
\end{center}}\medskip

Case \cut:\smallskip

{\small
\begin{center}
\AXC{$[(\DD_i,\gra_i\,;\,\grb_i)]_i$}
\RL{\scriptsize\weakex}\dashedLine
\UIC{$[(\GG_i,\DD_i,\gra_i\,;\,\grb_i)]_i$}
\RL{\scriptsize\impI}
\UIC{$[(\GG_i,\DD_i\,;\,\gra_i\rw\grb_i)]_i$}
\AXC{$[(\GG_i\,;\,\gra_i)]_i$}
\RL{\scriptsize\weak}\dashedLine
\UIC{$[(\GG_i,\DD_i\,;\,\gra_i)]_i$}
\RL{\scriptsize\impE}
\BIC{$[(\GG_i,\DD_i\,;\,\grb_i)]_i$}\DP
\end{center}}

 \end{proof}

\section{$\ISC$ and $\LJ$}
A derivation in $\ISC$ corresponds to a set of derivations in Intuitionistic Sequent Calculus ($\LJ$),
where the two conjunctions collapse into $\wedge$.

%applications
%of rules dealing with , which have an additive treatment of contexts.
%Moreover, one shall consider a function $\tilde{\_}$ from ISC derivations to families of LJ derivations:

\begin{dfn}\label{iusc-nj}
Given a $\ISC$ derivation $\pi$, we define the set $\tilde{\pi}=(\pi_i)_{i\in I}$ of $\LJ$
derivations by induction on the structure of $\pi$ as follows:\bigskip

$\bullet$ if\ \ \ $\infer[\scriptsize\ax]{\pi:[(\alpha_i;\alpha_i)\mid i\in I]}{}$
\,then\, $\tilde{\pi} = (\pi_i)_{i\in I}\,$ with\ \ \ $\infer[\LJax]{\pi_i:\alpha_i \vdash \alpha_i}{}$\bigskip

$\bullet$ if\ \ \ $\infer[\scriptsize\cut]{\pi:[(\Gamma_i, \Delta_i;\beta_i)\mid i\in I]}
{\pi^1:[(\Gamma_i;\alpha_i)\mid i\in I] \quad  \pi^2:[(\Delta_i, \alpha_i; \beta_i)\mid i\in I]}$
%with $\tilde{\pi^1}=(\infer[]{\Gamma_i \vdash \alpha_i}{\pi_i^1})_{i\in I}$ and $\tilde{\pi^2} =(\infer[]{\Delta, \alpha_i\vdash \beta_i}{\pi_i^2})_{i\in I}$
\,with\, $\tilde{\pi^1}=(\pi_i^1: \Gamma_i \vdash \alpha_i)_{i\in I}$\bigskip\\
and\, $\tilde{\pi^2} =(\pi_i^2: \Delta_i, \alpha_i\vdash \beta_i)_{i\in I}$,
then\, $\tilde{\pi}= (\pi_i)_{i\in I}$ \,with\ \ \ $\infer[\LJcut]{\pi_i:\Gamma_i, \Delta_i\vdash\beta_i}
{\pi^1_i:\Gamma_i \vdash \alpha_i \quad \pi^2_i:\Delta_i, \alpha_i\vdash \beta_i}$\bigskip

$\bullet$ if\ \ \ $\infer[\scriptsize\weak]{\pi:[(\Gamma_i, \alpha_i;\beta_i)\mid i\in I]}{\pi':[(\Gamma_i; \beta_i)\mid i\in I]}$ \,with\,
$\tilde{\pi'}= (\pi'_i:\GG_i\vdash\grb_i)_{i\in I}$, then\,
$\tilde{\pi}=(\pi_i)_{i\in I}$ \,with\ \ \ $\infer[\LJweak]{\pi_i:\Gamma_i, \alpha_i\vdash\beta_i}{\pi'_i:\Gamma_i\vdash \beta_i}$\bigskip

$\bullet$ $\tilde{\pi}$ is defined in the same way as above for \exch, \contr\ and \impR.\bigskip

%$\bullet$ if\ \ \ $\infer[\scriptsize (\bf P)]{\pi:[(\Gamma_i;\beta_i)\mid i\in I]}
%{[(\Gamma_i; \beta_i)\mid i\in I] \cup [(\Gamma_j; \beta_j)\mid j\in J] }$ \, then\,
%$\tilde{\pi}= \tilde{clean(\pi)}$
%\bigskip

$\bullet$ if\ \ \ $\infer[(\scriptsize \bf R)]{\pi:\mM}
{\mM' }$ \, where $(\scriptsize \bf R)=\FUS, (\scriptsize \bf P)$\, then\,
$\tilde{\pi}= \tilde{clean(\pi)}$
\bigskip

$\bullet$ if\ \ \ $\infer[\scriptsize\impL]
{\pi:[(\Gamma_i,\Delta_i,\alpha_i \rightarrow \beta_i ;  \gamma_i )\mid i \in I]}
{\pi^1:[(\Gamma_i ; \alpha_i ) \mid i \in I ] \quad \pi^2:[(\Delta_i,\beta_i ; \gamma_i) \mid i \in I ]}$
\,with\, $\tilde{\pi^1}= (\pi_i^1: \Gamma_i  \vdash \alpha_i)_{i\in I}$\bigskip\\
and\, $\tilde{\pi^2} = (\pi_i^2 : \Delta_i,\beta_i \vdash \gamma_i)_{i\in I}$,
then\, $\tilde{\pi}=(\pi_i)_{i\in I}$ \,with\ \ \ $\infer[\LJimpL]
{\pi_i:\Gamma_i,\Delta_i,\alpha_i \rightarrow \beta_i \vdash  \gamma_i}
{\pi^1_i:\Gamma_i \vdash \alpha_i \quad \pi^2_i:\Delta_i,\beta_i \vdash \gamma_i}$\bigskip

$\bullet$ if\ \ \ $\infer[\scriptsize\andL]
{\pi:[(\Gamma_i,\alpha^0_i \aconj \alpha^1_i ; \beta_i) \mid i \in I]}
{\pi':[(\Gamma_i,\alpha^0_i,\alpha^1_i ; \beta_i) \mid i \in  I ]}$
\,with\, $\tilde{\pi'}=(\pi'_i : \Gamma_i,\alpha^0_i,\alpha^1_i \vdash \beta_i)_{i\in I}$,
then\, $\tilde{\pi}=(\pi_i)_{i\in I}$\bigskip \\
with\ \ \ $\infer[\LJandL]{\pi_i:\Gamma_i,\alpha^0_i \LJand \alpha^1_i \vdash \beta_i}{\pi'_i:\Gamma_i,\alpha^0_i,\alpha^1_i \vdash \beta_i}$\bigskip

$\bullet$ if\ \ \ $\infer[\scriptsize\andR]
{\pi:[(\Gamma_i,\Delta_i ; \alpha_i \aconj \beta_i) \mid  i \in I]}
{\pi^1:[ (\Gamma_i; \alpha_i) \mid  i \in I ] \quad \pi^2:
[(\Delta_i; \beta_i )\mid  i \in I]}$
\,with\, $\tilde{\pi^1}=(\pi_i^1: \Gamma_i \vdash \alpha_i)_{i\in I}$\bigskip\\
and\,  $\tilde{\pi^2}=(\pi_i^2: \Delta_i \vdash \beta_i)_{i\in I}$,
then\, $\tilde{\pi} = (\pi_i)_{i\in I}$ \,with\ \ \ $\infer[\LJandR]
{\pi_i:\Gamma_i,\Delta_i \vdash  \alpha_i \LJand\beta_i}{\pi_i^1:\Gamma_i \vdash \alpha_i \quad \pi_i^2:\Delta_i \vdash \beta_i}$\bigskip

$\bullet$ if\ \ \ $\infer[\scriptsize\mbox{\bf($\sconj$L$_j$)}]
{\pi:[(\Gamma, \alpha_0 \sconj \alpha_1 ; \beta)] \cup  [(\Gamma_i ; \gamma_i) \mid i \in I]}
{\pi':[(\Gamma , \alpha_j ; \beta)] \cup  [(\Gamma_i ; \gamma_i) \mid i \in I]}$
\,with\, $\tilde{\pi'}= \{\pi'_\star: \Gamma , \alpha_j \vdash \beta\} \cup (\pi'_i: \Gamma_i \vdash \gamma_i)_{i\in I}$,\bigskip\\
then\, $\tilde{\pi}= \{\pi_\star\}\cup (\pi'_i)_{i\in I}$
\,with\ \ \ $\infer[\LJandL]{\pi_\star:\Gamma, \alpha_0 \LJand \alpha_1 \vdash \beta}{\infer[\scriptsize\weak,\exch]{\Gamma, \alpha_0,\alpha_1 \vdash \beta}{\pi'_\star:\Gamma,\alpha_j\vdash \beta}}$\bigskip

$\bullet$ if\ \ \ $\infer[\scriptsize\mbox{\bf ($\sconj$R)}]
{\pi:[(\Gamma ; \alpha \cap \beta)] \cup [(\Gamma_i ; \gamma_i) \mid i \in I]}
{\pi':[(\Gamma ; \alpha),(\Gamma ; \beta)] \cup [(\Gamma_i ; \gamma_i) \mid i \in I]}$
\,with\, $\tilde{\pi'}=\{\pi'_\alpha: \Gamma \vdash \alpha\}\cup\{\pi'_\beta: \Gamma \vdash \beta\}
\cup (\pi'_i: \Gamma_i \vdash \gamma_i)_{i \in I}$,\bigskip\\
then\, $\tilde{\pi}= \{\pi_\star\}\cup (\pi'_i)_{i\in I}$ \,with\ \ \ $\infer[\scriptsize\exch,\contr]{\pi_\star:\Gamma\vdash \alpha \LJand \beta}{\infer[\LJandR]
{\Gamma,\Gamma\vdash \alpha \LJand \beta}{\pi'_\alpha: \Gamma\vdash\alpha \quad \pi'_\beta:\Gamma\vdash\beta}}$\bigskip

\end{dfn}\bigskip

Let us stress the fact that, by Definition~\ref{iusc-nj}, each $\pi_i$ in $\tilde{\pi}$ is a derivation in $\LJ$. The
translation from $\ISC$ to $\LJ$ is almost standard, but for the rules $\FUS, \bf (P)$, where the result of Lemma
\ref{clean-lem} has been used.

\section{Cut-elimination in LJ: a short survey}
Analogously to $\LJ$, $\ISC$ enjoys cut-elimination, i.e., the cut rule is admissible.
In order to give a proof of this fact the most natural idea is to mimic the cut elimination procedure of $\LJ$,
in a parallel way. We will use Definition \ref{iusc-nj}, which associates to every derivation in
$\ISC$ a set of derivations in $\LJ$. 
There are different versions of such a proof in the literature; we suggest the versions in~\cite{gir,troel,omondi}.
Let us briefly recall it. \\

The cut-elimination algorithm in $\LJ$ is based on a definition of some elementary cut-elimination steps, each one depending on the
premises of the cut-rule, and on a certain order of applications of such steps. If the elementary steps are applied in the correct order,
then the procedure eventually stops and the result is a cut-free proof of the same sequent.

The principal problem for designing this algorithm is in defining the elementary step
when the contraction rule is involved. In fact the natural way
of defining the contraction step is the following:
$$
\infer[\scriptsize\mbox{\bf (cut)}]
{\Gamma, \Delta \vdash \beta}
{\Gamma \vdash\alpha \quad
\infer[\scriptsize\mbox{\bf (C)}]{\Delta, \alpha
 \vdash\beta}{\Delta, \alpha,\alpha \vdash\beta}}
 \smallskip\\
\qquad\rightarrow\quad\smallskip\\
\infer[\scriptsize\mbox{\bf (CX)}]{\Gamma, \Delta \vdash \beta}
{\infer[\scriptsize\mbox{\bf (cut)}]{\Gamma, \Gamma,\Delta \vdash \beta}
{\Gamma \vdash\alpha \quad
\infer[\scriptsize\mbox{\bf (cut)}] {\Gamma,\Delta, \alpha
 \vdash\beta}{\Gamma \vdash\alpha \quad
\Delta, \alpha,{\alpha}
 \vdash\beta }
 }
 }
$$
where $\footnotesize\mbox{\bf (CX)}$ represents a suitable number of contraction and exchange rules.
This step generates two cuts, such that the last one appears of the same ``size'' and also of the same
(or maybe greater) ``height'' as the original one, according to any reasonable notions of size and height.
A standard way to solve this problem is to strenghen the cut rule, allowing it to eliminate at the same
time more than one occurrences of a formula, in the following way:
$$\infer[\scriptsize\mbox{\bf (multicut)}]
{\Gamma, \Delta \vdash \beta}
{\Gamma \vdash\alpha \quad
\Delta, \overbrace{\alpha,...,\alpha}^m \vdash\beta \quad m\geq 1}$$
It is easy to check that replacing the cut-rule with the multicut-rule
produces an equivalent system. Abusing the naming, in what follows we will use the
name ``cut'' for multicut and we will call ``$\LJ$'' the system
obtained by replacing cut by multicut.

Assuming, for simplicity, that $\Gamma$ and $\Delta$ do not contain any
occurrence of $\alpha$, the contraction step now becomes:%\smallskip
$$
\infer[\scriptsize\mbox{\bf (cut)}]
{\Gamma, \Delta \vdash \beta}
{\Gamma \vdash\alpha \quad
\infer[\scriptsize\mbox{\bf (C)}]{\Delta, \alpha
 \vdash\beta}{\Delta, \alpha,\alpha \vdash\beta}}
 \smallskip\\
\qquad\rightarrow\qquad\smallskip\\
\infer[\scriptsize\mbox{\bf (cut)}]{\Gamma, \Delta \vdash \beta}
{\Gamma \vdash\alpha \quad
\Delta, \alpha,\alpha
 \vdash\beta
 }$$
and the new $cut$-rule which has been generated has been moved up in the derivation with respect to the original one.
The cut rule can be eliminated and, in order to design the algorithm doing it, the following measures are needed.

\begin{dfn}
\begin{itemize}
\item[i)] The {\em size} of a formula $\sigma$ (denoted by $|\sigma |$) is the number of symbols in it;
%\item[ii)] The {\em degree} of a cut is the size of the formula eliminated by it;
\item[ii)] The {\em height} of a derivation is the number of rule applications in its derivation tree. Let $h(\pi)$ denote the
height of $\pi$;
\item[iii)] The {\em measure} of a cut $\pi$, denoted by $m(\pi)$, is a pair ($s,h$), where $s$ is the size of the formula
eliminated by it and $h$ is the sums of the heights of its premises.
\end{itemize}
\medskip
We consider measures ordered according to a restriction of lexicographic order, namely:
 $$(s,h) < (s',h') \mbox{ if  either   }  s<s' \mbox{   and   } h\leq h' \mbox{   or   }
s \leq s' \mbox{   and   } h < h'.$$ 
\end{dfn}

Then the following lemma holds:
\begin{lem}\label{cut-elim-lem}
Let $\pi:\Gamma \vdash \sigma$ be a derivation in $\LJ$, with some cut rule applications. Let
$\pi':\Gamma \vdash \sigma$ be the derivation obtained from $\pi$ by applying an elementary cut-elimination step
to a cut closest to the axioms, let $c$.  Then $\pi'$ does not contain $c$ anymore, and, if it contains some new cuts
with respect to $\pi$,
their measure are less than the measure of
$c$. Moreover the measures of the cuts different from $c$ do not
increase.
\end{lem}

 The cut elimination property can now be stated.

\begin{thm}
Let $\pi$ be a derivation in $\LJ$. Then, there is a derivation $\pi'$ proving the same judgement which does not use
the cut-rule.
\end{thm}
\begin{proof}
From our definition of measure, a topmost cut can be eliminated in a finite number of steps. Since the number of cuts is finite,
the property follows.
\end{proof}

\section{Cut-elimination in $\ISC$: the elementary cut elimination steps}\label{el-step}

First of all we strenghten the cut-rule of $\ISC$ analogously to $\LJ$, in the following way:
$$\infer[\scriptsize\mbox{\bf (cut)}]
{[(\Gamma_i, \Delta_i; \beta_i) |  i \in I ]}
{[ (\Gamma_i; \alpha_i) | i \in I] \quad
[ (\Delta_i, \overbrace{\alpha_i,...,\alpha_{i}}^{m}; \beta_i) | i \in I] }$$
%In order to define a cut elimination algorithm for $\IUSC$, we need
%to start from the definition of the elementary cut-elimination steps,
%and then to define a cut-elimination algorithm using the correspondence defined in the previous section.
It is easy to check that we obtain an equivalent system, so we will abuse the notation, and call respectively
$cut$ the new cut, and $\ISC$ the new system.

Then, we can divide the most significant occurrences of cut in $\ISC$ in two cases: the global and local ones,
depending on whether global or local connectives are introduced on cut formulas in the premises.

In the case of global cut-rules, the elementary cut elimination steps act as for $\LJ$, but in parallel
on all the atoms of the involved molecules.
As an example, let us show the case of the {$(\aconj \scriptsize \bf R)$, $(\aconj \scriptsize \bf L)$} cut:\smallskip

\begin{center}
\AXC{$[(\GG_i;\gra_i)]_i$}
\AXC{$[(\DD_i;\grb_i)]_i$}
\RL{\scriptsize\andR}
\BIC{$[(\GG_i,\DD_i;\gra_i\aconj\grb_i)]_i$}
\AXC{$[(Z_i,\gra_i,\grb_i;\grg_i)]_i$}
\RL{\scriptsize\andL}
\UIC{$[(Z_i,\gra_i\aconj\grb_i;\grg_i)]_i$}
\RL{\scriptsize\cut}
\BIC{$[(\GG_i,\DD_i,Z_I;\grg_i)]_i$}\DP
\medskip\\
$\downarrow$\medskip\\
\AXC{$[(\GG_i;\gra_i)]_i$}
\AXC{$[(\DD_i;\grb_i)]_i$}
\AXC{$[(Z_i,\gra_i,\grb_i;\grg_i)]_i$}
\RL{\scriptsize\cut}
\BIC{$[(\DD_i,Z_i,\gra_i;\grg_i)]_i$}
\RL{\scriptsize\cut}
\BIC{$[(\GG_i,\DD_i,Z_i;\grg_i)]_i$}\DP
\end{center}\medskip

The definition of the local cut-elimination steps in $\IUSC$ poses some problems. Let us consider the following example:

Let $\pi_1$ and $\pi_2$ be respectively the following derivations:
$$
\infer[\scriptsize\intR]{\pi_1:[(\alpha, \alpha \ra \alpha; \alpha \sconj \alpha),(\mu, \mu \ra (\mu \sconj \nu); \mu \sconj \nu)]}
{\infer[\scriptsize\impL]{[(\alpha, \alpha \ra \alpha; \alpha),(\alpha, \alpha \ra \alpha; \alpha),(\mu, \mu\ra (\mu \sconj \nu); \mu \sconj \nu)]}
{\infer[\scriptsize\ax]{[(\alpha;\alpha),(\alpha;\alpha),(\mu;\mu)]}{}&\infer[\scriptsize\ax]{[(\alpha;\alpha),(\alpha;\alpha),(\mu \sconj \nu;\mu \sconj \nu)]}{}
}
}
$$
$$
\infer[\scriptsize\intL]{\pi_2:[(\alpha \sconj \alpha;\alpha \sconj \alpha), (\mu \sconj \nu; \mu)]}
{\infer[\scriptsize\ax]{[(\alpha \sconj \alpha;\alpha \sconj \alpha), (\mu ; \mu)]}{}}
$$
And let $\pi$ be:
$$
\infer[\scriptsize\cut]{\pi:[(\alpha, \alpha \ra \alpha; \alpha \sconj \alpha),(\mu, \mu \ra (\mu \sconj \nu); \mu)]}
{\pi_{1} \qquad\qquad \pi_{2}}
$$
Notice that the right and left introduction of the connective $\sconj$ create an asymmetry because they
are applied on different atoms, so it looks quite impossible to perform a cut-elimination step.
In fact, the derivation $\pi$ corresponds (modulo some structural rules) to the following derivation in $\ISL$:
$$
\infer[\scriptsize\intI]{[(\alpha, \alpha \ra \alpha; \alpha \sconj \alpha),(\mu, \mu \ra (\mu \sconj \nu); \mu) ]}
{\infer[\scriptsize\intE]{[(\alpha, \alpha \ra \alpha; \alpha),(\alpha, \alpha \ra \alpha; \alpha),(\mu,\mu \ra (\mu \sconj \nu); \mu )] }
{\infer[\scriptsize\impE]{[(\alpha, \alpha \ra \alpha; \alpha),(\alpha, \alpha \ra \alpha; \alpha),(\mu, \mu \ra (\mu \sconj \nu); \mu \sconj \nu)] }
{\infer[\scriptsize\ax]{[(\alpha \ra \alpha;\alpha \ra \alpha), (\alpha \ra \alpha;\alpha \ra \alpha), (\mu \ra (\mu \sconj \nu);\mu \ra (\mu \sconj \nu))]}{}
& \infer[\scriptsize\ax]{[(\alpha;\alpha), (\alpha;\alpha), (\mu;\mu)]}{}}}}
$$
This derivation is in normal form according to the definition given in~\cite{isl}; roughly speaking, a derivation in
$\ISL$ is in normal form if it does not contain an introduction of a connective immediately followed by an elimination
of the same connective, modulo some structural rules.

Since the problem is due to the presence of the local connective, in particular
when the two premises of a cut are the right and left introduction of $\sconj$ on different atoms,
the solution is to restrict the use of this connective,
in particular forbidding atoms where it is in principal position.

\begin{dfn}
\begin{itemize}
\item[i)] A formula is {\em canonical} if its principal connective is not $\sconj$.
\item[i)] An $\ISC$ derivation is {\em canonical} if all the formulas introduced by axioms and $(\scriptsize \bf W)$
are canonical.
\end{itemize}
\end{dfn}

\begin{lem}
Let $\pi: \ISCD \mM$. Then, there is a canonical derivation $\pi': \ISCD  \mM$.
\end{lem}

\begin{proof}
By induction on $\pi$.
\end{proof}\medskip

The following example illustrates the base-case.

\begin{exam}
Consider the derivation:
$$
\infer[\scriptsize\mbox{\bf (Ax)}]{[(\grs \sconj \grt; \grs \sconj \grt)]}{}$$
The corresponding canonical derivation can be obtained by introducing
the molecule $[(\grs;\grs),(\grt;\grt)]$ by an axiom and then by applying a sequence of two
$(\sconj L)$ rules followed by a $(\sconj R)$ rule.

\end{exam}

%The following lemma states a very important
%property of canonical derivations.

%\begin{lem}\label{local-comm}
%\begin{itemize}
%\item[i)] Rules $(\sconj R),(\sdis R)$ commute with all the left global rules.
%\item[ii)] Left and right local rules acting on the same atom commute each-other.
%\item[iii)] Local rules acting on different atoms commute each-other.
%\end{itemize}
%\end{lem}

From now on, we will only consider canonical derivations to define the cut-elimination steps.
Moreover we assume that the derivations are also clean.
As said before, the global rules are completely standard, since they act like in $\LJ$, but in parallel
on all the atoms of the involved molecules. Thus, we will expose only some characteristic cases dealing
with $\sconj$. Note the use of the structural rules ${\bf (P)}$ and $\FUS$.

\begin{dfn}\label{cut-step}
A cut-elimination step in $\ISC$ is defined by cases. We assume do not have applications of rules ${\bf (P)}$ and $\FUS$, and we show
the most interesting structural and local cases.

%{\bf Structural steps:}\medskip  %% degree unchanged, height decreases after eliminating (Fus)

%$\bullet$ \fbox{\bf case of \FUS}

%$$
%\begin{array}{c}
%\infer[\scriptsize\cut]{[(\Gamma_{i},\Delta_{i}; \delta_{i}), (\Gamma,\Delta; \delta), (\Gamma,\Delta; \delta)]}
%{\pi:\infer{[(\Gamma_{i};\gamma_{i}),(\Gamma;\delta),(\Gamma;\delta)]}{}
%&
%\infer[\scriptsize\FUS]{[(\Delta_{i},\gamma_{i};\delta_{i}),(\Delta,\gamma;\delta), (\Delta,\gamma;\delta)]}
%{\pi':[(\Delta_{i},\gamma_{i};\delta_{i}),(\Delta,\gamma;\delta)]}}
%\smallskip\\
%\downarrow\smallskip\\
%\infer[\scriptsize\FUS]{[(\Gamma_{i},\Delta_{i}; \delta_{i}), (\Gamma,\Delta; \delta), (\Gamma,\Delta; \delta)]}
%{
%\infer[\scriptsize\cut]{[(\Gamma_{i},\Delta_{i}; \delta_{i}), (\Gamma,\Delta; \delta)]}
%{
%\infer[\scriptsize (P)]{[(\Gamma_{i};\gamma_{i}),(\Gamma;\delta)]}
%{
%\pi:\infer{[(\Gamma_{i};\gamma_{i}),(\Gamma;\delta),(\Gamma;\delta)]}{}
%}
%&
%\infer{\pi':[(\Delta_{i},\gamma_{i};\delta_{i}),(\Delta,\gamma;\delta)]}{}
%}
%}
%\end{array}
%$$

{\bf Commutation steps:}\medskip  %% degree unchanged, height decreases after eliminating (Fus)

$\bullet$ \fbox{\bf case of \intR}

$$
\begin{array}
{c}
\infer[\scriptsize\cut]{[(\Gamma_0, \Delta_0; \alpha_0\cap\beta_0)] \cup [(\Gamma_i, \Delta_i; \delta_i)]_{1 \leq i \leq m}}
{\pi:[(\Gamma_0; \gamma_0)]\cup [(\Gamma_i; \gamma_i)]_{1 \leq i \leq m}
&
\infer[\scriptsize\intR]{\pi'':[(\Delta_0, \gamma_0; \alpha_0\cap\beta_0)] \cup [(\Delta_i, \gamma_i; \delta_i)]_{1 \leq i \leq m}}
{\pi':[(\Delta_0, \gamma_0; \alpha_0),(\Delta_0, \gamma_0; \beta_0)] \cup [(\Delta_i, \gamma_i; \delta_i)]_{1 \leq i \leq m}}}
\smallskip\\
\downarrow\smallskip\\
\infer[\scriptsize\intR]{[(\Gamma_0, \Delta_0; \alpha_0\cap\beta_0)] \cup [(\Gamma_i, \Delta_i; \delta_i)]_{1 \leq i \leq m}}
{\infer[\scriptsize\cut]{[(\Gamma_0, \Delta_0; \alpha_0),(\Gamma_0, \Delta_0; \beta_0)] \cup
[(\Gamma_i, \Delta_i; \delta_i)]_{1 \leq i \leq m}}
{
\infer[\scriptsize\FUS]{\pi''':[(\Gamma_0; \gamma_0),(\Gamma_0; \gamma_0)]\cup [(\Gamma_i; \gamma_i)]_{1 \leq i \leq m}}
{\pi:[(\Gamma_0; \gamma_0)]\cup [(\Gamma_i; \gamma_i)]_{1 \leq i \leq m}}
&
{\pi':[(\DD_0, \gamma_0; \alpha_0),(\DD_0, \gamma_0; \beta_0)] \cup [(\DD_i, \gamma_i; \delta_i)]_{1 \leq i \leq m}}
}}
\end{array}
$$\smallskip

{\bf Conversion steps:}\medskip

%\fbox{\bf case of symmetric local rules}\medskip

$\bullet$ \fbox{\bf case of symmetric $\sconj$ rule} %%% degree lower or the same, height decreases

$$
\begin{array}
{c}
\infer[\scriptsize\cut]{[(\Gamma_i, \Delta_i; \gamma_i)]_i}{
\infer[\scriptsize\intR]{[(\GG_0; \alpha_0 \cap \beta_0)] \cup [(\GG_i; \alpha_i)]_{1 \leq i \leq m}}
{\pi:[(\GG_0; \alpha_0), (\GG_0; \beta_0)] \cup [(\GG_i; \alpha_i)]_{1 \leq i \leq m}} &
\infer[\scriptsize\intL]{[(\DD_0, \alpha_0 \cap \beta_0; \gamma_0)] \cup [(\DD_i, \alpha_i; \gamma_i)]_{1 \leq i \leq m}}
{\pi':[(\DD_0, \alpha_0; \gamma_0)] \cup [(\DD_i, \alpha_i; \gamma_i)]_{1 \leq i \leq m}}}
\smallskip\\
\downarrow\smallskip\\
\infer[\scriptsize\cut]{[(\Gamma_i, \Delta_i; \gamma_i)]_i}{
\infer[\scriptsize\Pru]{[(\GG_0; \alpha_0)] \cup [(\GG_i; \alpha_i)]_{1 \leq i \leq m}}{\pi:[(\GG_0; \alpha_0),
(\GG_0; \beta_0)] \cup [(\GG_i; \alpha_i)]_{1 \leq i \leq m}} &
\pi':[(\DD_0, \alpha_0; \gamma_0)] \cup [(\DD_i, \alpha_i; \gamma_i)]_{1 \leq i \leq m}}
\\
\end{array}
$$\medskip

%\fbox{\bf case of asymmetric local rules}\medskip

$\bullet$ \fbox{\bf case of asymmetric $\sconj$ rule} %%% degree lower or the same, height decreases

$$
\begin{array}
{c}
\infer[\scriptsize\cut]{[(\GG_i,\DD_i;\grt_i)]_i\cup[(\GG,\DD;\grg),(\GG',\DD';\rho)]}
{\infer[\scriptsize\intR]{\pi:[(\GG_i;\grs_i)]_i\cup[(\GG;\gra \sconj \grb),(\GG'; \grm \sconj \grn)]}
{\pi_0:[(\GG_i;\grs_i)]_i\cup[(\GG;\gra),(\GG;\grb),(\GG';\grm \sconj \grn)]} & \infer[\scriptsize\intL]{\pi':[(\DD_i,\grs_i;\grt_i)]_i\cup
[(\DD,\gra \sconj \grb;\grg ),(\DD', \grm \sconj \grn;\rho)]}{\pi'_0:[(\DD_i,\grs_i;\grt_i)]_i
\cup [(\DD,\gra \sconj \grb;\grg ),(\DD',\grm; \rho)]}}
\end{array}
\medskip$$
Since the derivation is canonical, the $\sconj$ in the atom $(\GG'; \grm \sconj \grn)$ has been introduced by a
$(\sconj R)$ rule. Substituting this $(\sconj R)$ rule in $\pi$ by $(P)$ and removing the $(\sconj L)$ rule from $\pi'$,
we get:\medskip\medskip

\begin{center}
\AXC{$[\ \ldots,(\GG'';\grm),(\GG'';\grn)]$}
\RL{\scriptsize\intR}\UIC{$[\ \ldots,(\GG'';\grm\sconj\grn)]$}
\noLine
\UIC{\bf$\vdots$}\noLine
\UIC{\ \ \ }\noLine
\UIC{$[(\GG_i;\grs_i)]_i\cup[(\GG;\gra),(\GG;\grb),(\GG';\grm\sconj\grn)]$}
\RL{\scriptsize\intR}\UIC{$\pi:[(\GG_i;\grs_i)]_i\cup[(\GG;\gra\sconj\grb),(\GG';\grm\sconj\grn)]$}\DP
\quad$\rightarrow$\quad
\AXC{$[\ \ldots,(\GG'';\grm),(\GG'';\grn)]$}
\RL{\scriptsize\Pru}\UIC{$[\ \ldots,(\GG'';\grm)]$}
\noLine
\UIC{\bf$\vdots$}\noLine
\UIC{\ \ \ }\noLine
\UIC{$[(\GG_i;\grs_i)]_i\cup[(\GG;\gra),(\GG;\grb),(\GG';\grm)]$}
\RL{\scriptsize\intR}\UIC{$\pi_s:[(\GG_i;\grs_i)]_i\cup[(\GG;\gra\sconj\grb),(\GG';\grm)]$}\DP
\end{center}\medskip\medskip

\begin{center}
\AXC{$\pi_s:[(\GG_i;\grs_i)]_i\cup[(\GG;\gra\sconj\grb),(\GG';\grm)]$}
\AXC{$\pi'_0:[(\DD_i,\grs_i;\grt_i)]_i
\cup [(\DD,\gra \sconj \grb;\grg ),(\DD',\grm; \rho)]$}
\RL{\scriptsize\cut}
\BIC{$[(\GG_i,\DD_i;\grt_i)]_i\cup[(\GG,\DD;\grg),(\GG',\DD';\rho)]$}\DP
\end{center}\medskip\medskip

\end{dfn}

Note that a canonical derivation remains canonical after a cut-elimination step, while a clean derivation can be transformed into a not
clean one.
Morever, note that all the elementary steps related to global connectives and structural rules act locally.
In fact they correspond to apply the standard cut elimination steps of $\LJ$ in parallel in all the atoms of the considered
molecule. But some cases dealing with the local connective $\sconj$ are not local. In particular, in the asymmetric case,
the derivation can be modified in an almost global way.

%%%%%%%%%%%%%%%%%%%%%%%%%%%%%

\section{Cut-elimination in $\ISC$: the algorithm}

We are now ready to define the cut elimination algorithm for $\ISC$.
First we need to define a notion of measure, for proving the termination of the algorithm.
\begin{dfn}\label{cut-measure}
Let $\pi$ be a cut in $\ISC$, with premises $\pi_{1}$ and $\pi_{2}$.
The measure of $\pi$, denoted by $m(\pi)$ is the set
$\{ m(\pi_{i}) \mid \pi_{i} \in \tilde{\pi}\}$.
$m(\pi) \leq m(\pi')$ if and only if 
%either $m(\pi) \subseteq m(\pi')$ or $m(\pi)= M \cup M'$, $m(\pi')= M \cup M''$, $M' \cap M''=\emptyset$ and, 
for every $(s,h) \in m(\pi) \setminus m(\pi')$ there is $(s',h')\in m(\pi') \setminus m(\pi)$ such that $(s,h)\leq (s',h' )$.

%and let $\tilde{\pi}=(\pi_i)_{i\in I}$ its projection on $\LJ$, where all $\pi_{i}$
%end with a cut. Moreover, let $(s_{i}, h_{i})$ be the measure of $\pi_{i}$ ($i \in I$).
%and $ = \Sigma_{i\in I} s_{i}$, and $h= \Sigma_{i\in I} h_{i}$.
%We consider the measures ordered in lexicographic order.
\end{dfn}

%\begin{lem}\label{max-measure}
%Let $\pi:\mM$ be a canonical derivation in $\IUSC$, and let $m(\pi) \not=0$. Let
%$\pi':\mM$ be the derivation obtained by applying one elementary cut-elimination step
%to a cut-rule in $\pi$ of maximal degree and minimal height. Then $m(\pi')\leq m(\pi)$.
%\end{lem}
%\begin{proof}(sketch)
%Let $\mM=[(\Gamma_i , \alpha_i ; \beta_i) | i  \in I] $,
%$m(\pi)=(d,h)$, and let $\tilde{\pi}=(\pi_{i})_{\:i\in I}$ be the corresponding multiset
%in $\LJ$, according to the correspondence in Definition \ref{iusc-nj}.
%Let $m(\pi_{i})=(d_{i},h_{i})$.
%By Lemma \ref{max-cut}, there is at least one $\pi_{j}\in I$ such that this cut is maximal
%in $\pi_{j}$, in particular $m(\pi_{j})=(d,h+c)$
%and $m(\pi_{i})\leq m(\pi_{j})$, for all $i\in I$.
%Let $\pi'$ be the result of reducing the maximal cut in $\pi$.
%Then, by Lemma \ref{cut-step}), $m(\pi')\leq max_{i\in I}\{d_{i}\}$, so the proof is given.
%\end{proof}
%\medskip
An important lemma holds.
\begin{lem}\label{clean-measure}
Let $\pi$ be a derivation in $\ISC$, let $\tilde{\pi}=\{\pi_{i}\mid i\in I \}$
and let $\tilde{clean(\pi)}=\{\pi'_{i}\mid i\in I \} $. Then $h(\pi_{i}) = h(\pi'_{i})$, for all $\mid i\in I.$
\end{lem}
\begin{proof}
By Lemma \ref{clean-lem} and by the definition of $\tilde{\pi}$.
\end{proof}
\medskip

Now we are able to design the cut-elimination algorithm.

\begin{dfn}\label{alg}
\begin{itemize}
\item[i)] Let $\pi$ be a clean and canonical derivation in $\ISC$, containing at least one cut rule.
Then $step(\pi)$ is the result of applying to $\pi$ an elementary cut elimination step to a cut closest to the premises.
\item[ii)] The algorithm $\mathcal A$, which takes in input a clean and canonical proof $\pi$ and produces as output a
cut-free proof $\pi'$, is
defined as follows:
\begin{itemize}
\item ${\cal A}(\pi)= \textstyle if  \pi \mbox{ does not contain any cut then } \pi \mbox{ else } {\cal A}(clean(step(\pi)))$.
\end{itemize}
\end{itemize}
\end{dfn}

\begin{lem}\label{termin}
$\cal A(\pi)$ eventually stops.
\end{lem}
\begin{proof}
It would be necessary to check that every application of an elementary cut elimination step makes either the cut disappearing,
in case of an axiom cut, or replaces it by some cuts, with a less measure.  In case of a cut involving structural rules, or
global connectives, the check is easy, since the result comes directly from Lemma \ref{cut-elim-lem}, and by the definition
of measure of a cut in $\ISC$ which is done in function of the measure of a cut in $\LJ$ (remember that the input proof
is clean, so there are not occurrences of ($\scriptsize \bf P$) and of $(\FUS)$).
So we will show only the cases of the elementary steps defined in Definition \ref{cut-step}. For readability, we will use the same
terminology. \\

%Without loss of generality,  let us assume that, in both the premises of a cut, all atoms have different formulae on the
%right,
%so the cardinality of the molecule proved by a proof $\pi$ is the same as the number of elements of the set $\tilde{\pi}$.\\
%Moreover, let us denote by $H(\pi)$ the sum of the heights of all the derivations in $\tilde{\pi}$.\\

$\bullet$. {\bf commutation step ($\scriptsize \sconj \bf R$)}.\\

Let $\tilde{\pi}=\{ \pi_{0}:\Gamma_{0}\vdash \gamma_{0}\} \cup \{\pi_{i}: \Gamma_{i}\vdash \gamma_{i}\mid 1\leq i \leq m \}$,
$\tilde{\pi'}=\{ \pi^{1}_{0}:\Delta_{0},\gamma_{0}\vdash \alpha_{0}, \pi^{2}_{0}:\Delta_{0},\gamma_{0}\vdash \beta_{0}\} 
\cup \{\pi'_{i}: \Delta_{i}, \gamma_{i} \vdash \delta_{i }\mid 1\leq i \leq m \}$ and 
$\tilde{\pi''}=\{\pi'_{0}:\infer[(\sconj R)]{\Delta_{0},\gamma_{0}\vdash \alpha_{0}\sconj \beta_{0}}
{\pi^{1}_{0} & \pi^{2}_{0}} \} 
\cup \{\pi'_{i}: \Delta_{i}, \gamma_{i} \vdash \delta_{i }\mid 1\leq i \leq m \} $.
Moreover let $h_{0},h^{1}_{0},h^{2}_{0},h_{i},h'_{i}$ be respectively the heights of $\pi_{0},\pi^{1}_{0}, \pi^{2}_{0},\pi_{i}, \pi'_{i}$.

The measure of the cut is :
$m=\{(|\gamma_{0}|, h_{0 }+h^{1}_{0}+ h^{2}_{0} +1)\} \cup \{|\gamma_{i}|, h_{i}+h'_{i})\mid 1\leq i \leq m\}$.

The measure of the new generated cut is:
 $m'=\{(|\gamma_{0}|, h_{0 }+h^{1}_{0} ), (|\gamma_{0}|, h_{0 }+h^{2}_{0} )\} \cup \{|\gamma_{i}|, h_{i}+h'_{i})\mid 1\leq i \leq m\}$,
and $m' < m$, by Definition \ref{cut-measure}, remembering that $\tilde{\pi'''}=clean(\tilde{\pi'''})$ and 
since the $clean$ step does not modify the height of $\pi_{0}$
and $\pi_{i}$, by Lemma \ref{clean-measure}. \\

$\bullet$ {\bf Conversion step: symmetric $\sconj$ rule}.\\

Let $\tilde{\pi}= \{\pi^{1}_{0}:\Gamma_{0}\vdash \alpha_{0}, \pi^{2}_{0}:\Gamma_{0}\vdash \beta_{0}\}
\cup \{\pi_{i}: \Gamma_{i}\vdash \alpha_{i}\mid 1\leq i \leq m\}$ and
$\{\tilde{\pi'}=\pi_{0}^{3}:\Delta_{0},\alpha_{0}\vdash \gamma_{0}\}\cup \{\pi'_{i}:\Delta_{i},\alpha_{i}
\vdash \gamma_{i}\mid 1\leq i \leq m\}$.
Moreover let $h^{1}_{0},h^{2}_{0},h^{3}_{0},h_{i},h'_{i}$ be respectively the heights of 
$\pi_{0},\pi^{1}_{0}, \pi^{2}_{0},\pi^{3}_{0},\pi_{i}, \pi'_{i}$.

The measure of the cut is:
$m=\{(|\alpha_{0}\sconj \beta_{0}|,h^{1}_{0}+h^{2}_{0}+h^{3}_{0}+2 )\}\cup \{(|\alpha_{i}|,h_{i}+h'_{i} )\mid 1\leq i \leq m\}$.
The measure of the new generared cut is:
$m'=\{(|\alpha_{0}|,h^{1}_{0}+h^{3}_{0} )\}\cup \{(|\alpha_{i}|,h_{i}+h'_{i} )\mid 1\leq i \leq m\}$.
In computing $m'$, we used Lemma \ref{clean-measure}, which assures that the heigh of $\pi^{1}_{0}$
has been not modified by the $clean$ step.\\

$\bullet$ {\bf Conversion step: asymmetric $\sconj$ rule}\\

Let $\tilde{\pi_{0}}= \{\pi^{1}_{0}:\Gamma\vdash \alpha, \pi^{2}_{0}:\Gamma\vdash \beta, \pi^{3}_{0}:\Gamma' \vdash \mu \sconj \nu\}
\cup \{\pi_{i}: \Gamma_{i}\vdash \sigma_{i}\mid i\in I\}$ and
$\{\tilde{\pi'_{0}}=\pi_{0}^{4}:\Delta,\alpha \sconj \beta \vdash \gamma, \pi_{0}^{5}: \Delta',\mu \vdash \rho \}
\cup \{\pi'_{i}:\Delta_{i},\sigma_{i}
\vdash \tau_{i}\mid i\in I\}$.
Moreover let $h^{1}_{0},h^{2}_{0},h^{3}_{0},h^{4}_{0},h^{5}_{0},h_{i},h'_{i}$ be respectively the heights of 
$\pi_{0},\pi^{1}_{0}, \pi^{2}_{0},\pi^{3}_{0},\pi^{4}_{0},\pi^{5}_{0}\pi_{i}, \pi'_{i}$.

The measure of the cut is:
$m=\{(|\alpha\sconj \beta |,h^{1}_{0}+h^{2}_{0}+h^{4}_{0}+1 )\}\cup \{( |\mu \sconj \nu |, h^{3}_{0} + h^{5}_{0}+1) \}
\cup \{(|\sigma_{i}|,h_{i}+h'_{i} )\}$.
The measure of the new generared cut is:
$m'=\{(|\alpha \sconj \beta |,h^{1}_{0}+h^{2}_{0}+h^{4}_{0}+1 )\}\cup \{( |\mu |, h^{3}_{0} + h^{5}_{0}-1) \}
\cup \{(|\sigma_{i}|,h_{i}+h'_{i} )\}$, since, by Lemma \ref{clean-measure}, $h(\pi_{s})=h(\pi)$. \\
\\
So a topmost cut can be eliminated in a finite number of steps.
Moreover Lemma \ref{clean-measure} assures us
that the cleaning step does not
increase the measure of any cut.
Since the number of cuts is finite, the algorithm eventually stops.
\end{proof}

\begin{cor}
$\ISC$ enjoys the cut elimination property.
\end{cor}

\section{Intersection types from a logical point of view: an overview  }\label{Other}
The problem raised by Hindley, of looking for a logical system naturally connected to intersection type assignment through the Curry-Howard
isomorphism, has generated many different proposals.
Very roughly speaking, we could divide them in two categories,
that we call the {\em semantic approach} and the {\em logical approach} to the problem.
The semantic approach is characterized by the fact that an extension of the system we called $\IT$ in this paper has been considered.
Namely intersection type assignment comes with
a subtyping relation, which formalizes the semantics of intersection as the meet in the continuous function space.  This subtyping relation is the
essential tool in using intersection types for modelling denotational semantics of $\lambda$-calculus, and so we call this approach semantic.
The key idea here is to avoid the rule introducing the intersection, which has not a logical explication. The first result in this line is by
Venneri \cite{venneri94}, who designed an intersection type assignment system for Combinatory Logic, in Hilbert style, where
the introduction of
intersection is replaced by an infinite set of axioms for the basic combinators, built from their principal types. Then the subtyping
rule plays the role
of the intersection introduction. In \cite{venneri97} this result has been enhanced, both by extending it to union types, and by giving a logical
interpretation of the subtyping relation, which turned out to correspond to the implication in minimal relevant logic.
The connection between intersection types and relevant implication has been already noticed in \cite{barba94}.

In the logical approach the system shown $\IT$ shown in Fig. \ref{ITrules} is taken into consideration, without any subtyping relation.
The aim is to design a true deductive system, such that its decoration coincide with $\IT$.
In this line Capitani, Loreti and Venneri \cite{venneri01} propose a system of hypersequents, i.e., sequences
of formulae of $\LJ$, where the distinction between global and local connectives has been already introduced. The relation between this
system and $\LJ$ cannot be formally stated, since the notion of empty formula in an hypersequent is essential, while it has no correspondence
in $\LJ$. The relation between $\IT$ and $\LJ$ has been clarified by Ronchi Della Rocca and Roversi, who designed $\IL$ \cite{il},
a deductive system where formulae are tree of formulae of $\LJ$, proved by isomorphic proofs. The result has been further enhanced
by Pimentel, Ronchi Della Rocca and Roversi \cite{isl,islrev}, who defined $\ISL$ and proved that the intersection born from a splitting of the usual
intuitionistic conjunction into two connectives, each one reflecting part of its behaviour, local or global. The system $\ISL$ is extensively
discussed in this paper.

A problem strongly related to the considered one is the design of a language, explicitly
typed by intersection types.
Different proposals have been made.
In \cite{ronchi02}, a language with this property has been obtained as side effect of the logical approach, by a full decoration of the
intersection logic $\IL$. But its syntax is difficult, since the syncronous behaviour of intersection types is reflected
in the fact that a term typed by an intersection type is a pair of terms which are identical, modulo type erasure.
A similar language has been proposed in \cite{wells02a}. All the other attempts have been made with the aim of avoiding
such a duplication of terms.
Wells and Haack \cite{wells02} build a language where the duplication becomes dynamic,
by enriching the syntax and by defining an operation of type selection both on terms and types.
Liquori and Ronchi Della Rocca \cite{liquori07} proposed a language which has an imperative flavour, since terms are decorated
by locations, which in their turn contain intersection of simply typed terms, describing the corresponding type derivation.
The last proposal is by Bono, Bettini and Venneri \cite{venneri08}, and consists in a language with parallel features,
where parallel subterms share the same free variables, i.e., the same resources. Since we can see a connection between the
global behaviour of the arrow type and a parallel behaviour of terms, we think it would be interesting to explore if there
is a formal connection between this language and $\ISL$.

\bibliographystyle{eptcs} % or whatever you prefer

\begin{thebibliography}{1}
\bibitem{barba94} F. Barbanera, S. Martini,
  \newblock Proof-functional connectives and realizability
  \newblock \emph{Archive for Mathematical Logic}, 33, 189--211,
  \newblock 1994.
\bibitem
    {bar} F. Barbanera, M. Dezani-Ciancaglini, and U. de'Liguoro,
    \newblock Intersection and Union Types: Syntax and Semantics,
    \newblock \emph{Inform.\ and Comput.}, 119, pp.~202-230,
    \newblock 1995.
\bibitem{BCD} H. Barendregt, M. Coppo, and M. Dezani-Ciancaglini,
  \newblock A Filter Lambda Model and the Completeness of Type Assignment,
  \newblock \emph{Journal of Symbolic Logic}, 48(4), pp. 931-940,
  \newblock 1983.
  \bibitem{venneri08}
  V. Bono, L. Bettini and B. Venneri,
   \newblock A Typed lambda calculus with Intersection Types,
   \newblock \emph{Theoretical Computer Science}, vol. 398 (1-3), pp. 95-113,
   \newblock 2008.
\bibitem{venneri01} B. Capitani, M. Loreti, B. Venneri,
     \newblock Hyperformulae, Parallel Deductions and Intersection Types
     \newblock \emph{Electr. Notes Theor. Comput. Sci.}, vol. 50 (2),
     \newblock 2001.
\bibitem
    {it} M. Coppo and M. Dezani-Ciancaglini,
    \newblock An extension of the basic functionality theory for the lambda-calculus,
    \newblock \emph{Notre Dame J.\ Formal Logic}, 21(4), pp.~685-693,
    \newblock 1980.
  \bibitem{CDHL} Coppo M., Dezani-Ciancaglini M., Honsell F., and Longo
  G.  ``Extended Type Structures and Filter Lambda Models'', {\em
    Logic Colloquium '82}, pp.241-262, 1983.
    \bibitem
{venneri97} M. Dezani-Ciancaglini, S. Ghilezan, B. Venneri,
\newblock Intersection Types as Logical Formulae,
    \newblock \emph{Notre Dame Journal of Formal Logic}, vol. 38 (2), pp. 246-269,1997
       \newblock 1994.

%North Holland,
\bibitem
    {gir} J.-Y. Girard, Y. Lafont, and P. Taylor,
    \newblock Proofs and Types,
    \newblock \emph{Cambridge University Press},
    \newblock 1989.
\bibitem
    {hin} J.R. Hindley,
    \newblock Coppo-Dezani types do not correspond to propositional logic,
    \newblock \emph{Theoret.\ Comput.\ Sci.}, 28(1-2), pp.~235-236,
    \newblock 1984.

\bibitem{krivine} J.L.Krivine
\newblock Lambda-calcul, types et mod\`eles,
\newblock \emph{Masson},
\newblock 1990.

\bibitem{liquori07} L. Liquori, S. Ronchi Della Rocca,
     \newblock Intersection Types a la Church,
     \newblock \emph{Information and Computation}, vol. 205 (9), pp. 1371-1386,
      \newblock 2007.
    \bibitem
    {omondi} A.R. Omondi,
    \newblock Proof Normalization I: Gentzen Hauptsatz,
    \newblock \emph{Victoria University of Wellington}, Techn. Rep.CS-TR 93-5,
    \newblock 1993.
    \bibitem{Paolini}
L. Paolini, M. Piccolo, S. Ronchi Della Rocca
   \newblock Logical Semantics for Stability,
   \newblock \emph{proc. MFPS 2009}, LNCS, 249, pp. 429-449,
   \newblock 2009.

    \bibitem
    {isl} E. Pimentel, S. Ronchi Della Rocca, and L. Roversi,
    \newblock Intersection Types: a Proof-Theoretical Approach,
    \newblock \emph{ICALP'05 workshop, Proceedings of Structure and Deduction}, pp.~189-204,
    \newblock 2005.
\bibitem
    {islrev} E. Pimentel, S. Ronchi Della Rocca, and L. Roversi,
    \newblock Intersection Types from a proof-theoretic perspective,
    \newblock \emph{Fund. Informaticae: Special Issue on Intersection Types and Related Systems}, to appear.

\bibitem{pott80}
Pottinger, G.
\newblock A type assignment for the strongly normalizable $\lambda $-terms.
\newblock In {\em To H. B. Curry: essays on combinatory logic, lambda calculus
  and formalism}, pages 561-577. Academic Press, London, 1980.
  \bibitem{ronchi2} S. Ronchi Della Rocca and L. Paolini,
  \newblock The Parametric $\lambda$-Calculus: a Metamodel for Computation,
  \newblock \emph{Text in Theoretical Computer Science, Springer-Verlag},
  \newblock 2004.
\bibitem
    {il} S. Ronchi Della Rocca and L. Roversi,
    \newblock Intersection Logic,
    \newblock \emph{Proceedings of CSL'01}, LNCS 2142, pp.~414-428,
    \newblock 2001.
%\bibitem
%    {ced} S. Ronchi Della Rocca, A. Saurin, Y. Stavrinos, and A. Veneti,
%    \newblock A direct cut-elimination proof for IUSC, in preparation.
\bibitem{ronchi02} S. Ronchi Della Rocca,
\newblock Typed Intersection lambda calculus,
 \newblock \emph{Proc. of ITRS'02}, ENTCS, vol 70,
  \newblock 2002.
\bibitem
    {tut} Y. Stavrinos and A. Veneti,
    \newblock Towards a logic for union types,
    \newblock \emph{Fund. Informaticae: Special Issue on Intersection Types and Related Systems}, to appear.
\bibitem
   {troel} A.S. Troelstra and H. Schwichtenberg,
   \newblock \emph{Basic Proof Theory}, Cambridge University Press,
   \newblock 2000.
\bibitem
    {cie} A. Veneti and Y. Stavrinos,
    \newblock A Sequent Calculus for Intersection and Union Logic,
    \newblock \emph{CiE'08 Conference},
    abstract in \emph{Local Proceedings of the 4th CiE Conference}, p.~514,
    \newblock 2008.
\bibitem
    {venneri94} B. Venneri,
    \newblock Intersection Types as Logical Formulae,
    \newblock \emph{J.Log. Comput.}, vol. 4 (2), pp. 109-124,
       \newblock 1994.
\bibitem
    {wells02} J.B. Wells, C. Haack,
    \newblock Branching Types,
    \newblock \emph{proc. of ESOP'02}, LNCS, vol. 2305, pp. 115-132,
       \newblock 2002.
\bibitem{wells02a} J. B. Wells, Allyn Dimock, Robert Muller, and Franklyn Turbak,
\newblock  A calculus with polymorphic and polyvariant flow types,
\newblock \emph{ J. Funct. Programming}, 12(3), pp.183Ð227,
\newblock 2002.



\end{thebibliography}

%\newpage

%\input{comments_ISC.tex}

%\input{counterexample.tex}
%\input{permutations.tex}
%\input{commentsSec6.tex}
%\input{cut-counter-example.tex}
%\input{IUND-normalization}

\end{document}